\documentclass[11pt]{article}
\usepackage{authblk}
\usepackage[colorlinks,linkcolor=blue,anchorcolor=red,citecolor=red,pdfstartview=FitH]{hyperref}
\usepackage{graphics,graphicx}
\usepackage{mathtools}
\usepackage{amsmath,amssymb,amsthm,accents,amsfonts}
\usepackage{color}
\usepackage{cite}
\usepackage{dsfont}
\usepackage{bm}
\usepackage{rotating}

\def \b#1{\overline{#1}}
\def \t#1{\widetilde{#1}}
\def \h#1{\widehat{#1}}

\newtheorem{proposition}{Proposition}[section]

\numberwithin{equation}{section}
\allowdisplaybreaks

\begin{document}

\title{Darboux and binary Darboux transformations for discrete integrable systems \uppercase\expandafter{\romannumeral1}. Discrete potential KdV equation}
\author[1,2]{Ying Shi}
\author[2]{Jonathan J C Nimmo}
\author[1]{ Da-jun Zhang}
\affil[1]{Department of Mathematics, Shanghai University, Shanghai 200444,  P.R. China}
\affil[2]{Department of Mathematics, University of Glasgow, Glasgow G12 8QQ, UK}
\date{}
\maketitle

\begin{abstract}
The Hirota-Miwa equation can be written in `nonlinear' form in two ways: the discrete KP equation and, by using a compatible continuous variable, the  discrete potential KP equation.  For both systems, we consider the Darboux and binary Darboux transformations, expressed in terms of the continuous variable, and obtain exact solutions in Wronskian and Grammian form.  We discuss reductions of both systems to the discrete KdV and discrete potential KdV equations, respectively, and exploit this connection to find the Darboux and binary Darboux transformations and exact solutions of these equations.
\end{abstract}

\quad PACS numbers: 02.30.Ik, 05.45.Yv

\quad Mathematics Subject Classification: 39A10, 35Q58
\section{Introduction}

The Hirota-Miwa equation \cite{Hirota-1981, Miwa-1982} is the three-dimensional discrete integrable system
\begin{equation}\label{H-M-1}
(a_1-a_2)\tau_{12}\tau_3+(a_2-a_3)\tau_{23}\tau_1+(a_3-a_1)\tau_{31}\tau_2=0,
\end{equation}
where \emph{lattice parameters} $a_k$ are constants, $k=1,2,3$, and for $\tau=\tau(n_1,n_2,n_3)$ each subscript $i$ denotes a forward shift in the corresponding discrete variable $n_i$.  The Hirota-Miwa equation is well known as a three dimensional discrete master equation, which includes many other lower dimensional systems as reductions.  In this paper, we discuss in detail the reductions to the discrete Korteweg de Vries (dKdV) equation and the discrete potential KdV (dpKdV) equation. The dpKdV equation is also known as the H1 equation in the Adler-Bobenko-Suris (ABS) classification \cite{ABS-2002}.  This paper is the first one of a planned series which will explore the equations in the ABS list as reductions of the Hirota-Miwa equation, and then to exploit this connection to find the Darboux and binary Darboux transformations of these systems.

The Hirota-Miwa equation \eqref{H-M-1} can also be written in terms of `nonlinear' variables rather than $\tau$-function in two distinct ways\cite{Nimmo-2006,Gilson-Nimmo-2007}, when using variables $u^{ij}:=\tau_{ij}\tau/\tau_i\tau_j$ and the linear system
\begin{equation}\label{H-M-LP-u-form}
a_i\phi_i-a_j\phi_j=(a_i-a_j)u^{ij}\phi,\quad 1\leq i<j\leq3,
\end{equation}
where for $\phi=\phi(n_1,n_2,n_3)$ each subscript $i$ denotes a forward shift in the corresponding discrete variable $n_i$.  This linear system \eqref{H-M-LP-u-form} is compatible if and only if
\begin{subequations}\label{dKP-u-form}
\begin{align}
&(a_1-a_2)u^{12}+(a_2-a_3)u^{23}+(a_3-a_1)u^{13}= 0, \label{dKP-u-form-a}\\
&(u^{ij})_{_{k}} u^{ik}=(u^{ik})_{_{j}} u^{ij}. \label{dKP-u-form-b}
\end{align}
\end{subequations}
Each of the variables $u^{ij}$ relates to the solution $u=(\log \tau)_{xx}$ of the KP equation in the continuum limits, so we call \eqref{dKP-u-form-a} the discrete KP (dKP) equation.  Note that when one uses the formula $u^{ij}=\tau_{ij}\tau/\tau_i\tau_j$, \eqref{dKP-u-form-a} gives \eqref{H-M-1} and \eqref{dKP-u-form-b} is satisfied identically.  A second way is to suppose $u^{ij}=(v_j-v_i+(a_i-a_j))/(a_i-a_j)$, where $v=(\log \tau)_{x}$, so that $v$ satisfies the potential KP equation since $v_x=u$.  This ansatz solves \eqref{dKP-u-form-a} exactly and \eqref{dKP-u-form-b} becomes the discrete potential KP (dpKP) equation (see also \cite{Nijhoff-1984-dpkp}).

The outline of this paper is as follows.  In Section \ref{Hirota-Miwa}, we recall important definitions and properties of the Hirota-Miwa equation.  In particular, we write the Hirota-Miwa equation in `nonlinear' form in two ways: the discrete KP equation and, by using a compatible continuous variable, the discrete potential KP equation. For both equations, we give two different associated linear systems and their corresponding auxiliary linear systems in differential-difference form.  So their Darboux and binary Darboux transformations are given in \emph{differential}, rather than the more usual difference form \cite{Nimmo-1997}.  The differential form uses the first continuous flow $x$ of the KP hierarchy which is compatible with the discrete flows in the Hirota-Miwa equation. These transformations are used to derive exact solutions in Wronskian and Grammian form, respectively.  In Section \ref{dKdV equation}, we discuss the  2-reduction of the Hirota-Miwa equation, and in particular of the dKP to the dKdV equation and of the dpKP to the dpKdV equation (H1 in the ABS classification \cite{ABS-2002}). Then, by taking appropriate reductions of the results in Section \ref{Hirota-Miwa}, we derive the Lax pairs for the dKdV and dpKdV equations, and their Darboux and binary Darboux transformations and exact solutions.


\section{Hirota-Miwa Equation}\label{Hirota-Miwa}

\subsection{Wronskian and Grammian solutions of the Hirota-Miwa Equation}

The Hirota-Miwa equation \eqref{H-M-1} is the compatibility condition for the linear system \cite{Nimmo-1997, Nimmo-2006}
\begin{equation}\label{H-M-LP1}
a_i\phi_i-a_j\phi_j=(a_i-a_j)\frac{\tau_{ij}\tau}{\tau_i\tau_j}\phi,
\end{equation}
for $1\leq i<j\leq3$. It is invariant with respect to the reversal of all lattice directions $n_i\rightarrow -n_i$. On the other hand, the linear system \eqref{H-M-LP1} does not have such invariance and the reflections $n_i\rightarrow -n_i$ acting on \eqref{H-M-LP1} give a new linear system
\begin{equation}\label{H-M-LP2}
a_i\psi_{\b i}-a_j\psi_{\b j}=(a_i-a_j)\frac{\tau_{\b i\b j}\tau}{\tau_{\b i}\tau_{\b j}}\psi,
\end{equation}
for all $1\leq i<j\leq 3$. The subscript $\b i$ denotes a backward shift with respect to $n_i$, for example, $\psi_{\b 1}:=\psi(n_1-1,n_2,n_3)$. In \cite{Nimmo-1997, Nimmo-2006}, the difference form of the Darboux and binary Darboux transformations were derived for the Hirota-Miwa equation \eqref{H-M-1} and these were used to construct exact solutions in the form of Casoratian and discrete Grammian determinants. Here, we will express the Darboux and binary Darboux transformations in differential form instead, using the lowest order continuous flow $x$ of the KP hierarchy, and then the solutions obtained will be expressed as Wronskian and (continuous) Grammian determinants. The differential-difference linear equations for $\phi$ and $\psi$ are
\begin{equation}\label{H-M-LP1-x}
\phi_x=a_i\phi_i+\left(\left(\frac{\tau_x}{\tau}\right)_i-\frac{\tau_x}{\tau}-a_i\right)\phi,
\end{equation}
and
\begin{equation}\label{H-M-LP2-x}
\psi_x=-a_i\psi_{\b i}+\left(\left(\frac{\tau_x}{\tau}\right)_{\b i}-\frac{\tau_x}{\tau}+a_i\right)\psi,
\end{equation}
where the subscript $x$ denotes the derivative, with $\tau$ satisfying the semi-discrete KP equation
\begin{equation}\label{semi-discrete-KP}
(a_i-a_j)(\tau_{ij}\tau-\tau_i\tau_j)+\tau_{i,x}\tau_j-\tau_{i}\tau_{j,x}=0.
\end{equation}
It is straightforward to check that \eqref{H-M-LP1-x} and \eqref{H-M-LP2-x} are compatible with \eqref{H-M-LP1} and \eqref{H-M-LP2}.  Note that the reflection symmetry which relates \eqref{H-M-LP1} and \eqref{H-M-LP2} may be extended by adding $x\to-x$ to relate \eqref{H-M-LP1-x} and \eqref{H-M-LP2-x}.

The basic Darboux transformation for the Hirota-Miwa equation is stated in the following proposition.
\begin{proposition}\label{prop1-HM}
Let $\theta$ be a non-zero solution of the linear system \eqref{H-M-LP1} and \eqref{H-M-LP1-x} for some $\tau$. Then the transformations
\begin{equation}\label{H-M-DT1-1-step}
\phi\rightarrow\t \phi=a_i(\phi_i-\theta_i\theta^{-1}\phi)=\phi_x-\theta_x\theta^{-1}\phi, \quad \tau\rightarrow\t\tau=\theta\tau,
\end{equation}
leave \eqref{H-M-LP1} and \eqref{H-M-LP1-x} invariant, where $i=1,2,3$.
\end{proposition}

The proof is a straightforward computation.  Note that there are four expressions for $\t\phi$ in the Darboux transformation \eqref{H-M-DT1-1-step}.  These are equivalent because of the linear equations \eqref{H-M-LP1} and \eqref{H-M-LP1-x}.

Next we write down the formula for iterated Darboux transformations, which give solutions in Wronskian and Casoratian determinant form.  The Wronskian determinant is the determinant of the $N\times N$ matrix with columns $\bm\Theta^{^{(j)}}\!=\!\bm\Theta^{^{(j)}}\!(x,y,t)$, for $j=0,1,\dots,N-1$, where $\bm\Theta^{^{(0)}} \!=\!({\theta_1}(x,y,t), {\theta_2}(x,y,t), \dots, {\theta_N}(x,y,t))^T$ and $\bm\Theta^{^{(j)}}\!=\! \partial^j\bm\Theta^{(0)}/\partial x^j$. It is written as
$$W({\theta_1}, {\theta_2}, \dots, {\theta_N})=|\bm\Theta^{(0)},\bm\Theta^{(1)},\dots,\bm\Theta^{(N-1)}|,$$
or in a more compact notation
$$W({\theta_1}, {\theta_2}, \dots, {\theta_N})=|\h{N-1}|.$$
The Casoratian determinant can be seen as a discrete analogue of the Wronskian determinant.  It is the determinant of the $N\times N$ matrix with columns $\bm\Theta^{^{(j)}}\!=\!\bm\Theta^{^{(j)}}\!(n_1,n_2,n_3)$, for $j=0,1,\dots,N-1$, where $\bm\Theta^{^{(0)}}\!=\!({\theta_1}(n_1,n_2,n_3), {\theta_2}(n_1,n_2,n_3), \dots, {\theta_N}(n_1,n_2,n_3))^T$, and $\bm\Theta^{(j)}$ is defined by the forward shifts, i.e.  $\bm\Theta^{(j)}=T_{n_1}^j(\bm\Theta^{(0)})=\bm\Theta(n_1+j,n_2,n_3)$. It is written as
\begin{align}\label{Casoratian-1}
C(\theta_1,\theta_2, \dots, \theta_N)=|\bm\Theta^{(0)},\bm\Theta^{(1)},\dots,\bm\Theta^{(N-1)}|=|0,1,\dots,N-1|=|\h{N-1}|.
\end{align}
We can also use $\bm\Theta^{(\b j)}$, which is defined by the backward shifts, to replace the $\bm\Theta^{(j)}$, where $\bm\Theta^{(\b j)}=T_{\b{n_1}}^j(\bm\Theta^{(0)})=\bm\Theta(n_1-j,n_2,n_3)$.

Below, we use the subscript $[i]$ to designate that the shifts of the Casoratian determinant are with respect to the variable $n_i$. For example, the Casoratian determinant in \eqref{Casoratian-1} could be denoted as $C_{_{[1]}}(\theta_1,\theta_2, \dots, \theta_N)$.

\begin{proposition}\label{prop1N-HM}
Let ${\theta_1}, {\theta_2}, \dots, {\theta_N}$ be non-zero, independent solutions of the linear system \eqref{H-M-LP1} and \eqref{H-M-LP1-x} for some $\tau$. Then $N$ applications of the above Darboux transformations give the transformations
\begin{equation}\label{N-DT1}
\phi\rightarrow\phi[N]=\frac{a_i^N C_{_{[i]}}(\theta_1,\theta_2, \dots, \theta_N,\phi)}{C_{_{[i]}}(\theta_1,\theta_2, \dots, \theta_N)}=\frac{W(\theta_1,\theta_2, \dots, \theta_N,\phi)}{W(\theta_1,\theta_2, \dots, \theta_N)},
\end{equation}
and
\begin{equation}\label{N-DT1-2}
\tau\rightarrow\tau[N]=a_i^{\frac{N(N-1)}{2}} C_{_{[i]}}(\theta_1,\theta_2, \dots, \theta_N)\tau= W(\theta_1,\theta_2,\dots, \theta_N)\tau,
\end{equation}
which leave \eqref{H-M-LP1} and \eqref{H-M-LP1-x} invariant. Here $C_{[i]}$ denotes the Casoratian determinant in  forward shifts with respect to the discrete variable $n_i$, for $i=1,2,3$.
\end{proposition}

For example, by using the results \eqref{N-DT1-2}, the $N$-soliton solutions of the Hirota-Miwa equation can be expressed in both Casoratian and Wronskian form in terms of the eigenfunctions,
\begin{equation}\label{theta-k}
\theta_k(n_1,n_2,n_3)=e^{p_{\!_k}\!x}\prod^{3}_{i=1}(1+\frac{p_{_k}}{a_i})^{n_i} +e^{q_{\!_k}\!x}\prod^{3}_{i=1}(1+\frac{q_{_k}}{a_i})^{n_i}.
\end{equation}
Here $\theta_k$, $k=1,2,\dots, N$, is obtained from \eqref{H-M-LP1} and \eqref{H-M-LP1-x}, by choosing the trivial solution $\tau=1$.

Now we can apply the reflections $n_i \rightarrow -n_i$ and $x\rightarrow -x$ to the above results to deduce corresponding result for the second linear system \eqref{H-M-LP2} and \eqref{H-M-LP2-x}.
\begin{proposition}\label{prop2-HM}
Let $\rho$ be a non-zero solution of the linear system \eqref{H-M-LP2} and \eqref{H-M-LP2-x} for some $\tau$. Then the transformations
\begin{equation}\label{H-M-DT2-1-step}
\psi\rightarrow\t \psi=a_i(\psi_{\b i}-\rho_{\b i}\rho^{-1}\psi)=\psi_x-\rho_x\rho^{-1}\psi, \quad \tau\rightarrow\t\tau=\rho\tau,
\end{equation}
leave \eqref{H-M-LP2} and \eqref{H-M-LP2-x} invariant, for all $i=1,2,3$.
\end{proposition}

In the statement of this proposition the linearity of \eqref{H-M-LP2} and  \eqref{H-M-LP2-x} allows us to omit a minus sign.

\begin{proposition}\label{prop2N-HM}
Let ${\rho_1}, {\rho_2}, \dots, {\rho_N}$ be non-zero, independent solutions of the linear system \eqref{H-M-LP2} and \eqref{H-M-LP2-x} for some $\tau$.  Then $N$ applications of the above Darboux transformations give the transformations
\begin{equation}
\psi\rightarrow\psi[N]=\frac{a_i^N C_{_{[\b i]}}(\rho_1,\rho_2, \dots, \rho_N,\psi)}{C_{_{[\b i]}}(\rho_1,\rho_2, \dots, \rho_N)}=\frac{W(\rho_1,\rho_2, \dots, \rho_N,\psi)}{W(\rho_1,\rho_2, \dots, \rho_N)},\label{N-DT2N-1}
\end{equation}
and
\begin{equation}
\tau\rightarrow\tau[N]=a_i^{\frac{N(N-1)}{2}} C_{_{[\b i]}}(\rho_1,\rho_2, \dots, \rho_N)\tau= W(\rho_1,\rho_2, \dots, \rho_N)\tau,\label{N-DT2N-2}
\end{equation}
which leave \eqref{H-M-LP2} and \eqref{H-M-LP2-x} invariant. Here $C_{[\b i]}$ denotes the Casoratian determinant in backward shifts with respect to the discrete variable $n_i$, for $i=1,2,3$.
\end{proposition}

The $N$-soliton solutions of the Hirota-Miwa equation are given by \eqref{N-DT2N-2}. From the linear system \eqref{H-M-LP2} and \eqref{H-M-LP2-x}, if we choose the seed solution as $\tau=1$, then the Casoratian and Wronskian determinants are defined by eigenfunctions,
\begin{equation}\label{rho-k}
\rho_k(n_1,n_2,n_3)=e^{-p_kx}\prod^{3}_{i=1}(1+\frac{p_{_k}}{a_i})^{-n_i}+e^{-q_kx}\prod^{3}_{i=1} (1+\frac{q_{_k}}{a_i})^{-n_i},
\end{equation}
for $k=1,2,\dots,N$.

To construct a binary Darboux transformation, we introduce the potential $\omega=\omega(\phi, \psi)$, defined by the relations
\begin{align}
\Delta_i\omega(\phi, \psi)&=\phi \psi_i, \label{omega-1}\\
\omega_x(\phi, \psi)&=\phi\psi,\label{omega-2}
\end{align}
where $\Delta_i=a_i(T_{n_i}-1)$ and $T_{n_i}$ is the forward shift operator in variable $n_i$, for $i=1, 2, 3$.  If $\phi$ and $\psi$ satisfy the linear systems \eqref{H-M-LP1}, \eqref{H-M-LP1-x} and \eqref{H-M-LP2}, \eqref{H-M-LP2-x}, respectively, then \eqref{omega-1} and \eqref{omega-2} are compatible. So the potential $\omega$ is well-defined.

The following proposition gives the binary Darboux transformation of the Hirota-Miwa equation.

\begin{proposition}\label{prop12-HM}
Suppose $\theta$ and $\phi$ are non-zero solutions of the linear system \eqref{H-M-LP1} and \eqref{H-M-LP1-x}, $\rho$ and $\psi$ are non-zero solutions of the linear system \eqref{H-M-LP2} and \eqref{H-M-LP2-x}, then the transformations
\begin{align}
\phi\rightarrow\h\phi&=\phi-\theta w(\theta, \rho)^{-1}w(\phi,\rho),\label{H-M-bDT-1}\\
\psi\rightarrow\h\psi&=\psi-\rho w(\theta, \rho)^{-1}w(\theta,\psi),\label{H-M-bDT-2}
\end{align}
leave \eqref{H-M-LP1}, \eqref{H-M-LP1-x} and \eqref{H-M-LP2}, \eqref{H-M-LP2-x} invariant, respectively, with $\tau$ changing to
\begin{equation}\label{H-M-bDT-1-tau}
\h\tau=w(\theta, \rho)\tau.
\end{equation}
\end{proposition}

The $N$-fold iteration of the binary Darboux transformation is as follows.

\begin{proposition}\label{prop12N-HM}
Let $\theta_1,\dots,\theta_N$ and $\rho_1,\dots,\rho_N$ satisfy the linear systems \eqref{H-M-LP1}, \eqref{H-M-LP1-x} and \eqref{H-M-LP2}, \eqref{H-M-LP2-x}, respectively. Define $N$-vectors $\bm\Theta=(\theta_1,\dots,\theta_N)^T$ and $\bm P=(\rho_1,\dots,\rho_N)^T$. Then $N$ applications of the binary Darboux transformation give the transformations
\begin{align}
\phi\rightarrow\phi[N]=\begin{vmatrix}
\bm\Omega(\bm\Theta, \bm P) & \bm\Theta\\
\bm\Omega(\phi, \bm P) & \phi \\
\end{vmatrix}
|\bm\Omega(\bm\Theta, \bm P) |^{-1},
\quad
\psi\rightarrow \psi[N]=\begin{vmatrix}
\bm\Omega(\bm\Theta, \bm P) &\bm P \\
\bm\Omega( \bm\Theta,\psi) & \psi \\
\end{vmatrix}
|\bm\Omega(\bm\Theta, \bm P) |^{-1},\label{H-M-N-bDT-1}
\end{align}
which leave \eqref{H-M-LP1}, \eqref{H-M-LP1-x} and \eqref{H-M-LP2}, \eqref{H-M-LP2-x} invariant, respectively, with $\tau$ changing to
\begin{equation}
\tau[N]=|\bm\Omega(\bm\Theta, \bm P)|\tau.\label{H-M-N-bDT-2}
\end{equation}
Here $\bm\Omega(\bm\Theta, \bm P)=(\omega(\theta_i, \rho_j))_{i,j=1,\dots,N}$ is an $N\times N$ matrix, $\bm\Omega(\phi, \bm P)=(\omega(\phi,\rho_j))_{j=1,\dots,N}$ and $\bm\Omega(\psi, \bm \Theta)=(\omega(\theta_i,\psi))_{i=1,\dots,N}$ are $N$-row vectors.
\end{proposition}

The proofs of those above propositions are straightforward computation, so we do not give the details. The reader is also referred to the papers \cite{Nimmo-1997,Nimmo-2006}.

\subsection{The discrete potential KP equation}\label{dpkp equation}

By introducing a potential $v(n_1,n_2,n_3;x):=\tau_x/\tau$, the semi-discrete equation \eqref{semi-discrete-KP} gives the relation
\begin{equation}\label{trans.-dpKP-dKP}
u^{ij}=\frac{\tau_{ij}\tau}{\tau_{i}\tau_j}=\frac{v_j-v_i+(a_i-a_j)}{a_i-a_j}, \quad 1\leq i<j\leq 3,
\end{equation}
for $v$ each subscript $i$ denotes a forward shift in the corresponding discrete variable $n_i$. So \eqref{dKP-u-form-a} is satisfied identically, and \eqref{dKP-u-form-b} becomes
\begin{subequations}\label{dpkp}
\begin{equation}\label{dpkp-1}
\frac{v_2-v_1+a_1-a_2}{(v_2-v_1+a_1-a_2)_{_{3}}}=\frac{v_3-v_1+a_1-a_3}{(v_3-v_1+a_1-a_3)_{_{2}}}=\frac{v_3-v_2+a_2-a_3}{(v_3-v_2+a_2-a_3)_{_{1}}}.
\end{equation}
The equation \eqref{dpkp-1} can be written in two other equivalent forms as either
\begin{equation}\label{dpkp-2}
\left(v_2-v_1+a_1-a_2\right)\left(v_{23}-v_{12}+a_1-a_3\right)=\left(v_3-v_1+a_1-a_3\right) \left(v_{23}-v_{13}+a_1-a_2\right),
\end{equation}
or
\begin{equation}\label{dpkp-3}
(a_3\!+\!v_{12})\!\left(a_1\!-\!a_2\!+\!v_2\!-\!v_1\right)\!+\!(a_2\!+\!v_{13})\!\left(a_3\!-\!a_1\!+\! v_1\!-\!v_3\right)\!+\!(a_1+v_{23})\!\left(a_2\!-\!a_3\!+\!v_3\!-\!v_2\right)=0.
\end{equation}
\end{subequations}
Equation \eqref{dpkp} is called the discrete potential KP (dpKP) equation. It was first found in \cite{Nijhoff-1984-dpkp} as the nonlinear superposition of solutions to the potential KP equation related by the B\"{a}cklund transformations, cf.\cite{Nijhoff-Notes-1} as well.  The trivial solution of the dpKP equation \eqref{dpkp} could be $v=c$, where $c$ is an arbitrary constant.

The linear systems associated with the dpKP equation, obtained by using relations \eqref{H-M-LP1} and  \eqref{H-M-LP2}, together with \eqref{trans.-dpKP-dKP}, are
\begin{equation}\label{dpKP-LP-1}
a_i\phi_i-a_j\phi_j=(v_j-v_i+a_i-a_j)\phi,
\end{equation}
and
\begin{equation}\label{dpKP-LP-2}
a_i\psi_{\b i}-a_j\psi_{\b j}=(v_{\b j}-v_{\b i}+a_i-a_j)\psi,
\end{equation}
where $1\leq i<j\leq3$. Its corresponding differential-difference linear
systems are \eqref{H-M-LP1-x} and \eqref{H-M-LP2-x} with $\tau_x/\tau=v$.

Together with the differential-difference linear system \eqref{H-M-LP1-x}, the Darboux transformation of the linear system \eqref{dpKP-LP-1} gives the new solution of the dpKP equation
\begin{equation*}
\t v=(\log(\theta\tau))_x=v+(\log\theta)_x,
\end{equation*}
where $\theta$ is a non-zero solution of \eqref{dpKP-LP-1} and \eqref{H-M-LP1-x}.  More generally, $N$-fold iteration gives the Wronskian solution
\begin{equation*}
v[N]=v+(\log W(\theta_1,\theta_2,\dots,\theta_N))_x,
\end{equation*}
where $\theta_k$, $k=1,2,\dots, N$ are the non-zero independent solutions of \eqref{dpKP-LP-1} and \eqref{H-M-LP1-x}.

The binary Darboux transformation gives a new solution of the dpKP equation
\begin{equation*}
\h  v=v+(\log\omega(\theta,\rho))_x,
\end{equation*}
where $\omega(\theta,\rho)$ is defined by \eqref{omega-1} and \eqref{omega-2}, and $\theta$, $\rho$ are non-zero solutions of \eqref{dpKP-LP-1}, \eqref{H-M-LP1-x} and \eqref{dpKP-LP-2}, \eqref{H-M-LP2-x}, respectively. Its $N$-fold iteration gives the Grammian solution
\begin{equation*}
v[N]=v+(\log|\bm \Omega(\bm\Theta,\bm P)|)_x,
\end{equation*}
where the $N\times N$ matrix $\bm \Omega(\bm\Theta,\bm P)=(\omega(\theta_i, \rho_j))_{i,j=1,\dots,N}$, $\omega(\theta_i,\rho_j)$ is defined by \eqref{omega-1} and \eqref{omega-2}, and $\theta_i$, $\rho_j$ are non-zero solutions of \eqref{dpKP-LP-1}, \eqref{H-M-LP1-x} and \eqref{dpKP-LP-2}, \eqref{H-M-LP2-x}, respectively.

\section{Reductions to the dKdV and dpKdV equations}\label{dKdV equation}

\subsection{The dKP equation to the dKdV equation}

The discrete KdV (dKdV) equation is a 2-reduction of the dKP equation \eqref{dKP-u-form-a}. In the reduction, it is necessary that one takes $a_2+a_3=0$, since the solutions satisfy the reduction condition $\tau_{23}=\tau$, cf.\cite{Kajiwara-Ohta-2008}. The reduction condition for the eigenfunction of the linear system \eqref{H-M-LP1} and \eqref{H-M-LP1-x} is $\phi_{23}=(1-\lambda^2)\phi$, where the form of the coefficient is chosen for its correspondence with the discrete one-dimensional Schr\"{o}dinger equation \cite{ Case-1973, Willox-2010}. To realise this reduction, we make the ansatz
\begin{eqnarray}\label{reduction}
\tau(n_1,n_2,n_3) &=&\tau(n_1,n_2-n_3)\label{reduction-1}\\
\phi(n_1,n_2,n_3)&=&(1-\lambda^2)^{n_3}\phi(n_1,n_2-n_3). \label{reduction-2}
\end{eqnarray}
Now taking the reduction conditions, the dKP equation \eqref{dKP-u-form-a} can be written as either
\begin{subequations}\label{H-M-rewrite}
\begin{align}
(a_2-a_1)\frac{\tau_{12}\tau}{\tau_1\tau_2}+(a_2+a_1)\frac{\tau_{1\b2}\tau}{\tau_1\tau_{\b2}} &=2a_2\frac{\tau}{\tau_{\b2}}\frac{\tau}{\tau_2},\label{H-M-rewrite-1}\\
\intertext{or}
(a_2-a_1)\frac{\tau_{1}\tau_{\b2}}{\tau_{1\b2}\tau}+(a_2+a_1)\frac{\tau_{1}\tau_2}{\tau_{12}\tau} &=2a_2\frac{\tau_1}{\tau_{1\b2}}\frac{\tau_1}{\tau_{12}}.\label{H-M-rewrite-2}
\end{align}
\end{subequations}
We then express the above two equations in terms of nonlinear variable
\begin{equation}\label{u-redu0}
u(n_1,n_2):=\frac{\tau_1\tau_3}{\tau_{13}\tau}=\frac{\tau_{1}\tau_{\b 2}}{\tau_{1\b 2}\tau}.
\end{equation}
By eliminating the tau-function parts on the right hand sides in both \eqref{H-M-rewrite-1} and \eqref{H-M-rewrite-2}, we obtain
\begin{equation}\label{dKdV}
\frac{1}{u_1}-\frac{1}{u_2}=\frac{a_1-a_2}{a_1+a_2}\left (u_{12}-u\right),
\end{equation}
which is the discrete KdV equation, first found by Hirota \cite{Hirota-1977}.

\begin{subequations}\label{dKdV-LP1}
From the linear system \eqref{H-M-LP1} and using the 2-reduction, we get the linear system of the dKdV equation
\begin{align}
a_1\phi_1-a_2\phi_2&=(a_1-a_2)u_2\phi,\label{dKdV-LP1-1}\\
a_1\phi_1+a_2(1-\lambda^2)\phi_{\b 2}&=(a_1+a_2)\frac{1}{u}\phi.\label{dKdV-LP1-2}
\end{align}
Note here that, by eliminating the $\phi_1$ in these two equations, we have
\begin{equation*}
a_2\phi_2+a_2(1-\lambda^2)\phi_{\b2}=\left((a_2-a_1)u_2+(a_2+a_1)\frac{1}{u}\right)\phi,
\end{equation*}
which is a discrete one-dimensional Schr\"{o}dinger equation \cite{Case-1973, Willox-2010}.
\end{subequations}

\begin{subequations}\label{dKdV-LP2}
The dKdV equation \eqref{dKdV} is also invariant with respect to the  reflections $n_i\rightarrow -n_i$, for $i=1,2$. So applying the reflections to the system \eqref{dKdV-LP1} gives a new linear system of the dKdV equation
\begin{align}
a_1\psi_{\b1}-a_2\psi_{\b2}&=(a_1-a_2)u_{\b2}\psi,\\
a_1\psi_{\b1}+a_2(1-\lambda^2)\psi_{2}&=(a_1+a_2)\frac{1}{u}\psi.
\end{align}
\end{subequations}
This system could also be obtained from the linear system \eqref{H-M-LP2}, by using the 2-reduction $\tau_{\b2\b3}=\tau$, and $\psi_{\b2\b3}=(1-\lambda^2)\psi$ with $a_2+a_3=0$.

For the linear differential-difference equations \eqref{H-M-LP1-x} and \eqref{H-M-LP2-x}, the 2-reduction does not affect the $n_1$- or $n_2$- parts, but the $n_3$-parts become
\begin{align}
\phi_x&=-a_2(1-\lambda^2)\phi_{\b 2}+\left(\left(\frac{\tau_x}{\tau}\right)_{\b2}-\frac{\tau_x}{\tau}+a_2\right)\phi,\label{dKdV-LP1-x}\\
\psi_x&=a_2(1-\lambda^2)\psi_{2}+\left(\left(\frac{\tau_x}{\tau}\right)_{2}-\frac{\tau_x}{\tau}-a_2\right)\psi.\label{dKdV-LP2-x}
\end{align}\

The fundamental Darboux transformation for the dKdV equation is
\begin{proposition}\label{prop4}
\begin{subequations}\label{dKdV-DT1}
Suppose $\theta$ is a non-zero solution of the linear system \eqref{dKdV-LP1} and \eqref {dKdV-LP1-x} for some $u$ and $\tau$,  together with $u=\tau_{1}\tau_{\b 2}/\tau_{1\b 2}\tau$, then the transformations
\begin{equation}
\phi\rightarrow\t \phi=a_i(\phi_i-\theta_i\theta^{-1}\phi)=\phi_x-\theta_x\theta^{-1}\phi, \quad i=1,2,
\end{equation}
and
\begin{equation}
\tau\rightarrow \t\tau=\theta\tau, \quad u\rightarrow \t u=u\frac{\theta_1\theta_{\b2}}{\theta_{1\b2}\theta},
\end{equation}
leave \eqref{dKdV-LP1} and \eqref{dKdV-LP1-x} invariant.
\end{subequations}
\end{proposition}

\begin{proof}
One way to complete the proof is a straightforward computation. A second way is using the idea that the Darboux transformation of the dKP equation in the Proposition \ref{prop1-HM} also works for the dKdV equation after taking the 2-reduction.  From the linear system \eqref{dKdV-LP1}, $\phi$ and $\theta$  are its solutions, so we have
\begin{equation}\label{dKdV-DT1-1}
a_i(\phi_i-\theta_i\theta^{-1}\phi)=-a_2(1-\lambda^2)(\phi_{\b2}-\theta_{\b2}\theta^{-1}\phi),
\end{equation}
for $i=1,2$. On the other hand, the 2-reduction gives
\begin{equation}\label{dKdV-DT1-2}
a_3(\phi_3-\theta_3\theta^{-1}\phi)=-a_2(1-\lambda^2)(\phi_{\b2}-\theta_{\b2}\theta^{-1}\phi).
\end{equation}
For the potential $u$, we have
$$\t u=\frac{\t\tau_{1}\t\tau_{\b 2}}{\t\tau_{1\b 2}\t\tau}=u\frac{\theta_1\theta_{\b2}}{\theta_{1\b2}\theta}.$$
So together with the transformation in  \eqref{H-M-DT1-1-step}, it means that after taking 2-reduction, the transformations \eqref{dKdV-DT1} leave the linear system \eqref{dKdV-LP1} and \eqref{dKdV-LP1-x} invariant.
\end{proof}

Similarly, by using the reflections $n_i\rightarrow -n_i$ and $x\rightarrow -x$, or the 2-reduction of the Darboux transformation of the dKP equation in Proposition \ref{prop2-HM}, we get the Darboux transformation of the linear system \eqref{dKdV-LP2} and \eqref{dKdV-LP2-x}.
\begin{proposition}\label{prop5}
\begin{subequations}\label{dKdV-DT2}
Suppose $\rho$ is a non-zero solution of the linear system \eqref{dKdV-LP2} and \eqref {dKdV-LP2-x} for some $u$ and $\tau$,  together with $u=\tau_{\b1}\tau_{2}/\tau_{\b12}\tau$, then the transformations
\begin{gather}
\psi\rightarrow\t \psi=a_i(\psi_{\b i}-\rho_{\b i}\rho^{-1}\psi)=\psi_x-\rho_x\rho^{-1}\psi, \quad i=1,2,\\
\intertext{and}
\tau\rightarrow \t\tau=\rho\tau, \quad u\rightarrow \t u=u\frac{\rho_{\b1}\rho_{2}}{\rho_{\b12}\rho},
\end{gather}
leave \eqref{dKdV-LP2} and \eqref {dKdV-LP2-x} invariant.
\end{subequations}
\end{proposition}

Similarly, in the light of the reductions, the binary Darboux transformations of the dKdV equation can also be obtained directly from the one in Proposition \ref{prop12-HM}. Thus the results of the $N$-applications of the Darboux and binary Darboux transformations for the dKdV equation can be gotten from the ones in Propositions \ref{prop1N-HM}, \ref{prop2N-HM} and \ref{prop12N-HM}. For this reason, we will not go to talk about them in detail here.

\subsection{The dpKP equation to the dpKdV equation}

Here again we use the 2-reduction, $v_{23}=v$ and $a_2+a_3=0$. Now by introducing the potential $v(n_1,n_2;x):=\tau_x/\tau$ into \eqref{trans.-dpKP-dKP}, we get
\begin{subequations}\label{tran-dkdv-dpkdv}
\begin{eqnarray}
(a_1-a_2)u_2 &=& v_2-v_1+a_1-a_2, \label{tran-dkdv-dpkdv-1}\\
(a_1+a_2)\frac{1}{u} &=& v_{\b 2}-v_1+a_1+a_2. \label{tran-dkdv-dpkdv-2}
\end{eqnarray}
\end{subequations}
By eliminating the potential $u$ from \eqref{tran-dkdv-dpkdv}, we obtain
\begin{equation}\label{dpKdV}
(v_2-v_1+a_1-a_2)(v-v_{12}+a_1+a_2)=a_1^2-a_2^2.
\end{equation}
This is the discrete potential KdV (dpKdV) equation \cite{Nijhoff-Notes-2,Nijhoff-2009}, and we see that its solution can be written as $v=(\log\tau)_x$. The dpKdV equation could also be obtained by the permutability property of the B\"{a}cklund transformations of the continuous potential KdV equation,  cf.\cite{Nijhoff-Notes-2}.  Taking the potential transformation $v\rightarrow v+a_1n_1+a_2n_2+\gamma$, where $\gamma$ is an arbitrary constant, \eqref{dpKdV} becomes
\begin{equation}\label{H1}
(v_2-v_1)(v-v_{12})=a_1^2-a_2^2,
\end{equation}
which is called the H1 equation\cite{ABS-2002}.

By taking the 2-reduction, the linear system of the dpKP equation \eqref{dpKP-LP-1} and \eqref{dpKP-LP-2} give the Lax pairs of the dpKdV equation \eqref{dpKdV},
\begin{subequations}\label{H1-LP1}
\begin{eqnarray}
a_1\phi_1-a_2\phi_2&=&(v_2-v_1+a_1-a_2)\phi,\\
a_1\phi_1+a_2(1-\lambda^2)\phi_{\b 2}&=&(v_{\b2}-v_1+a_1+a_2)\phi,
\end{eqnarray}
\end{subequations}
and with reflections $n_i\rightarrow -n_i$, this also gives
\begin{subequations}\label{H1-LP2}
\begin{eqnarray}
a_1\psi_{\b1}-a_2\psi_{\b2}&=&(v_{\b2}-v_{\b1}+a_1-a_2)\psi,\\
a_1\psi_{\b1}+a_2(1-\lambda^2)\psi_{2}&=&(v_{2}-v_{\b1}+a_1+a_2)\psi.
\end{eqnarray}
\end{subequations}
Their corresponding differential-difference linear systems are \eqref{H-M-LP1-x} and \eqref{H-M-LP2-x} with $\tau_x/\tau=v$, for $i=1,2$, respectively. Note here that the Lax pair \eqref{H1-LP1} coincides with the one derived from the multidimensional consistency property, cf.\cite{Nijhoff-Notes-2,Nijhoff-2009,Nalini-2010}, by using appropriate parameters transformations, such as $r=a_2\lambda $.

As we showed for the dpKP equation, the Darboux transformation gives the new solution of the dpKdV equation in differential form,
\begin{equation*}
\t v=v+(\log\theta)_x,
\end{equation*}
where $\theta$ is the non-zero solution of \eqref{H1-LP1} and \eqref{H-M-LP1-x}. Its $N$-fold iteration gives the solution in Wronskian form
\begin{equation} \label{N-soliton-dpKdV}
v[N]=v+(\log W(\theta_1,\theta_2,\dots,\theta_N))_x,
\end{equation}
where $\theta_k$, for $k=1,2,\dots,N$, are the non-zero solutions of \eqref{H1-LP1} and \eqref{H-M-LP1-x}.

The soliton solutions of the dpKdV equation \eqref{dpKdV} can be obtained from \eqref{N-soliton-dpKdV}. From the linear system \eqref{H1-LP1} and \eqref{H-M-LP1-x}, by  choosing the seed solution $v=0$, the eigenfunctions of the $\tau$-function in Wronskian determinant are
\begin{equation*}
\theta_k(n_1,n_2)=e^{a_2\lambda_{\!_k}x}(1+\lambda_k\beta)^{n_1}(1+\lambda_k)^{n_2}
+e^{-a_2\lambda_{\!_k}x}(1-\lambda_k\beta)^{n_1}(1-\lambda_k)^{n_2}, \quad \beta=\frac{a_2}{a_1},
\end{equation*}
for $k=1,2,\dots,N$. This result coincides with the one given by Hietarinta and Zhang \cite{JH-DJZ-2009}. In their paper,  they employed $f=|\h{N-1}|$ and $g=|\h{N-2},N|$, which were given as Casoratian determinants with respect to an auxiliary discrete variable. This auxiliary variable is compatible with the original independent variables, $n_1$ and $n_2$, but is external to the system. We will denote this auxiliary variable to be $n_4$ here. Then for an arbitrary constant $c$, $v=\dfrac{g}{f}+c$ satisfies the dpKdV equation \eqref{dpKdV}. $f$ and $g$ in Casoratian determinants in compact form are
\begin{align*}
f&=|\h{N-1}|=|\phi, T_{n_4}(\phi),T_{n_4}^2(\phi),\dots,T_{n_4}^{N-1}(\phi)|,\\
\intertext{and}
g&=|\h{N-2},N|=|\phi, T_{n_4}(\phi),T_{n_4}^2(\phi),\dots,T_{n_4}^{N-2}(\phi),T_{n_4}^{N}(\phi)|,
\end{align*}
where $\phi=\phi(n_1,n_2, n_4,x)$ and $T_{n_4}^j(\phi)=\phi(n_1,n_2, n_4+j,x)$, $j=0,1,\dots,N$. Similarly, $\tau$ and $\tau_x$ in Wronskian determinants are
\begin{align*}
\tau&=|\h{N-1}|=|\phi, \partial_x\phi,\partial_x^2\phi,\dots,\partial_x^{N-1}\phi|,\\
\intertext{and}
\tau_{x}&=|\h{N-2},N|=|\phi, \partial_x\phi,\partial_x^2\phi,\dots,\partial_x^{N-2}\phi,\partial_x^{N}\phi|.
\end{align*}
So for the linear systems \eqref{H1-LP1} and \eqref{H-M-LP1-x} with $\tau_x/\tau=v$, we could also introduce a virtual discrete variable $n_4$, say, which is another compatible discrete flow, with parameter $a_4$. If we choose the seed solution of the dpKdV equation $v=0$, then the differential-difference linear system \eqref{H-M-LP1-x}, with respect to the discrete variable $n_4$, gives the entries of the above Wornsikan and Carsoratian determinants satisfying the relations
\begin{align*}
\phi_x&=a_4(T_{n_4}-1)(\phi),
\intertext{and}
\partial_x^{N}\phi&=a_4^N(T_{n_4}-1)^N(\phi).
\end{align*}
So we have
\begin{equation*}
\tau=a_4^{\frac{N(N-1)}{2}}f,\quad \tau_{x}=a_4^{\frac{N(N-1)}{2}+1}(g-Nf).
\end{equation*}
By taking $a_4=1$, the soliton solution of the dpKdV equation is
\begin{equation*}
v=(\log \tau)_x=\frac{\tau_x}{\tau}=\dfrac{g}{f}-N.
\end{equation*}
So if the soliton solution of the dpKdV equation is expressed in Wronskian determinant, by using a compatible continuous variable, then it needs to employ $\tau$ and its derivative $\tau_x$. If it is expressed in Casoratian determinant, through a virtual discrete variable $n_4$, then it needs to employ $f$ and $g$. But there is no direct relation between $f$ and $g$.

The binary Darboux transformation gives the exact solution of the dpKdV equation,
\begin{equation*}
\h  v=v+(\log\omega(\theta,\rho))_x,
\end{equation*}
where $\omega(\theta,\rho)$ is defined by \eqref{omega-1} and \eqref{omega-2}, and $\theta$, $\rho$ are non-zero solutions of \eqref{H1-LP1}, \eqref{H-M-LP1-x} and \eqref{H1-LP2}, \eqref{H-M-LP2-x}, respectively, with $\tau_x/\tau =v$. Its $N$-fold iteration gives the exact solution in Grammian form,
\begin{equation*}
v[N]=v+(\log|\bm \Omega(\bm\Theta,\bm P)|)_x,
\end{equation*}
where the $N\times N$ matrix $\bm \Omega(\bm\Theta,\bm P)=(\omega(\theta_i, \rho_j))_{i,j=1,\dots,N}$,  $\omega(\theta_i,\rho_j)$ is defined by \eqref{omega-1} and \eqref{omega-2}, and $\theta_i$, $\rho_j$ are non-zero solutions of \eqref{H1-LP1}, \eqref{H-M-LP1-x} and \eqref{H1-LP2}, \eqref{H-M-LP2-x}, respectively, with $\tau_x/\tau =v$.

\section{Conclusions}
In this paper, we write the Hirota-Miwa equation in `nonlinear' form in two ways: the discrete KP equation in terms of the variable $u^{ij}=\tau_{ij}\tau/\tau_i\tau_j$ and, by using a compatible continuous variable, the discrete potential KP equation in variable $u^{ij}=(v_j-v_i+(a_i-a_j))/(a_i-a_j)$, where $v=(\log \tau)_{x}$.  This leads us to see clearly the relationship between the dKP and dpKP equations, which is similar to the relationship $u=v_x$ in the continuous case. For both systems, by introducing differential-difference linear systems, we get the Darboux and binary Darboux transformations and exact solutions in Wronskian and Grammian form.  In the second part of the paper, we discuss in detail the 2-reduction of Hirota-Miwa equation, and in particular, of the dKP to the dKdV equation and of the dpKP to the dpKdV equation. Then by taking appropriate reductions of the results for the dKP and dpKP equations, respectively, we derive the Lax pairs for the dKdV and dpKdV equations, and hence their Darboux and binary Darboux transformations and exact solutions.

\section*{Acknowledgements}
One of the authors (Y. Shi) is a PhD candidate supported by a fund under the State Scholarship Fund organised by the China Scholarship Council (CSC). She would also like to thank Duncan McCoy for helpful discussions.  This project is partially supported by the NSF of China (No. 11071157), the SRF of the DPHE of China (No. 20113108110002) and the Project of ``First-class Discipline of Universities in Shanghai''.


\end{document}